\documentclass[fleqn,11pt,letterpaper,english]{article}
\usepackage[T1]{fontenc}
\usepackage[latin9]{inputenc}
\usepackage[english]{babel}
\usepackage[a4paper]{geometry}
\usepackage{color}
\usepackage{babel}

\usepackage{float}
\usepackage{textcomp}
\usepackage{latexsym}
\usepackage{eurosym}
\usepackage{url}
\usepackage{amsthm}
\usepackage{amsmath}
\usepackage{graphicx}
\usepackage{amssymb}
\usepackage{esint}
\usepackage{xcolor}

\makeatletter
\numberwithin{equation}{section}
\numberwithin{figure}{section}
\theoremstyle{plain}
\newtheorem{thm}{Theorem}
  \theoremstyle{plain}
  \newtheorem{prop}[thm]{Proposition}
   \newtheorem{definition}[thm]{Definition}
  
  \theoremstyle{plain}
  \newtheorem{assumption}[thm]{Assumption}
  \theoremstyle{remark}

  \newcommand{\alias}[2]{
\providecommand{#1}{}
\renewcommand{#1}{#2}
}

\alias{\P}{\mathbb{P}}
\alias{\N}{\mathcal{N}}
\alias{\L}{\mathcal{L}}
\alias{\Z}{\mathbb{Z}}
\alias{\Q}{\mathbb{Q}}
\alias{\R}{\mathbb{R}}
\alias{\C}{\mathcal{C}}
\alias{\T}{\mathbb{T}}
\alias{\E}{\mathbb{E}}
\alias{\H}{\mathcal{H}}
\alias{\1}{\mathbf{1}}
\alias{\B}{\mathcal{B}}
\alias{\M}{\mathcal{M}}
\alias{\G}{\mathcal{G}}
\alias{\Y}{Y_{\bullet}}

\def\Esp{\mathbb{E}}

\usepackage{ifpdf} 
\ifpdf 

 \IfFileExists{lmodern.sty}{\usepackage{lmodern}}{}

\fi 

\let\myTOC\tableofcontents
\renewcommand\tableofcontents{%
  \frontmatter
  \pdfbookmark[1]{\contentsname}{}
  \myTOC
  \mainmatter }

\def\LyX{\texorpdfstring{%
  L\kern-.1667em\lower.25em\hbox{Y}\kern-.125emX\@}
  {LyX}}

\@addtoreset{footnote}{section}

\@ifundefined{showcaptionsetup}{}{%
 \PassOptionsToPackage{caption=false}{subfig}}
\usepackage{subfig}
\makeatother

\author{Ren\'e A\"id\footnote{Université Paris-Dauphine, PSL Research University, LEDa.}, \hspace{1.5mm} Luciano Campi\footnote{Corresponding author; Statistics Department, London School of Economics. Postal address: Statistics Department, Columbia House, London School of Economics and Political Science, 10 Houghton Street, London WC2A2AE. Email: l.campi@lse.ac.uk}, \hspace{1.5mm} Delphine Lautier\footnote{Universit\'e Paris Dauphine, PSL Research University, DRM UMR CNRS 7088, and FiME Lab - Laboratoire de Finance des March\'es d'Energie. 
}}
\title{ On the spot-futures no-arbitrage relations \\ in commodity markets \thanks{This study was supported by the Finance and Sustainable Development Chair sponsored by EDF and CACIB, and hosted by Ecole Polytechnique, Université Paris-Dauphine, PSL Research University and CREST.}}

\begin{document}

\maketitle

\begin{abstract}
In commodity markets the convergence of futures towards spot prices, at the expiration of the contract, is usually justified by no-arbitrage arguments. In this article, we propose an alternative approach that relies on the expected profit maximization problem of an agent, producing and storing a commodity while trading in the associated futures contracts. In this framework, the relation between the spot and the futures prices holds through the well-posedness of the maximization problem. We show that the futures price can still be seen as the risk-neutral expectation of the spot price at maturity and we propose an explicit formula for the forward volatility. Moreover, we provide an heuristic analysis of the optimal solution for the production/storage/trading problem, in a Markovian setting. This approach is particularly interesting in the case of energy commodities, like electricity: this framework indeed remains suitable for commodities characterized by storability constraints, when standard no-arbitrage arguments cannot be safely applied.  \\

Keywords: futures contracts; no-arbitrage relation; commodity; production; storage; energy. 

JEL Codes: C32; G13; Q4; Q02. 
\end{abstract}
\newpage

\section{Introduction}
In this article we aim at explaining, through a parsimonious model, the relation between the spot and the futures prices of a commodity. The prices relation derives from  the well-posedness of the optimization problem of an operator involved in the production and (when possible) the storage of a commodity. This producer also trades in the futures market. Such an approach is different from the classical no-arbitrage reasoning usually employed to explain the temporal basis in commodity markets. Interestingly, it remains relevant for non storable commodities, when standard no-arbitrage arguments cannot be safely applied\footnote{This does not mean, however, that our framework is relevant only to non-storable commodities, like electricity, nor that it is focused on commodities with storable inputs (A\"id \emph{et al}, 2009 and 2013 \cite{Aid09,Aid13}).}. 

The traditional no-arbitrage relation, initially developed for investment assets like stocks and bonds, states that the futures price $F(t,T)$ of a contract written on the investment asset, observed at date $t$ for a delivery at $T$, is the spot price $S_t$ capitalized at the interest rate $r$ between $t$ and $T$: $F(t,T) = S_t e^{r(T-t)}$. The cost of carrying commodities, however, can not be reduced to the interest rate. One need to take into consideration warehousing and/or depreciation costs (see, among others, Fama and French, 1987 \cite{Fama_French_87} or, more recently, Eydeland, 2002 \cite{Eydeland02}, Chap. 4, pp. 140-43). 

Introducing storage costs in the analysis is not even sufficient to depict the behavior of the prices spread in commodity markets: in addition, there is a need to explain negative spreads\footnote{A positive spread (a contango) arises when at date $t$, the futures price $F(t,T)$ is higher than the spot price $S_t$. A negative spread (a backwardation) is a situation where  $F(t,T) < S_t$. } and the simultaneous presence of positive inventories and backwardated prices (see among others, Working 1933 \cite{Working33}). The no-arbitrage relation can only be preserved by the introduction of an extra variable: the convenience yield (see Kaldor, 1940 \cite{Kaldor40} or more recently, Lautier, 2009 \cite{Lautier09}). The latter represents an implicit revenue. The presence of such a yield, which is associated with the holding of the physical commodity but not with the futures contract, explains that the operators maintain their inventories even in the presence of negative spreads. With the introduction of this variable, the no arbitrage relation is extended to:  $F(t,T) = S_t e^{(r- y)(T-t)}$, where $y$ is the convenience yield net of storage costs.

Financial models explicitly designed for the pricing of commodity derivatives and relying on the concept of convenience yield have been quite extensively built in the literature (see Lautier, 2005 \cite{Lautier_05} for a review). These models perform well, especially the two-factor model proposed by Schwartz (1997 \cite{Schwartz97}), where the convenience yield is mean reverting and plays the role of a dividend yield in the drift of the spot price. In this setting, however, the relation between the spot and the futures prices relies on the hypothesis that arbitrage operations are perfect. Consequently, there exists a unique risk-neutral measure, that can be calibrated thanks to market data. Moreover the spot price, in such analysis, is exogenous.\\

An effort has been undertaken in the economic and financial literature in order to propose structural models that rely on endogenous spot prices to explain the interaction between the spot and futures prices. In such cases, the spot price derives from production, consumption and storage of the commodity. In this setting, the spot price is the result of an equilibrium between production and consumption. The futures price is defined as the expectation of the spot price. \\
Within this literature, our model is close to those of Brennan (1958 \cite{Brennan58}) and Routledge \emph{et al} (2000 \cite{Routledge00}). In these works, the authors develop production-storage models that connect the spot and the futures prices. Nevertheless, the models are mainly developed in order to allow for comparative statics and little information on the conditions under which no-arbitrage holds is given (this is especially true for the model of Routledge \emph{et al} (2000 \cite{Routledge00}). \\
The development of electricity markets in the last thirty years has introduced new challenges in the literature on commodity prices. More precisely, the idea that the futures price of a commodity is linked to its spot price by a convenience yield has been highly debated, as electricity cannot be stored. Many authors have argued that the convenience yield may not apply in the case of electricity. Still, futures prices of electricity exhibit both contango and backwardation situations (Benth \emph{et al}, 2013 \cite{Benth13}). Others have also stressed that the no-arbitrage method used in mathematical finance to obtain a risk-neutral measure should be reconsidered (Benth \emph{et al}, 2003 \cite{Benth03}). Finally, due to the particular nature of the futures contracts negotiated in these markets (which are basically swaps, see Frestad, 2010 \cite{Frestad10}), even the convergence of the futures price to the spot price has become an issue (Viehmann, 2011 \cite{Viehmann11}).\\
Since in the case of electricity the pricing of derivatives cannot rely on the concept of a convenience yield - at least in its restricted initial definition, where the yield comes from the holding of inventories - this non-storable commodity has fostered research on how to restore a relation between the spot and the futures prices. The first approaches relied on two-date equilibrium models (Anderson and Hu, 2008 \cite{Anderson08}, Bessembinder and Lemmon, 2002 \cite{Bessembinder02}, Aid \emph{et al}, 2011 \cite{Aid11}), where the risk-neutral measure was extracted from the risk aversion parameters of the agents. A more complex approach consists in the extension of the market beyond the underlying asset, to include or production factors (A\"id \emph{et al}, 2009 and 2013 \cite{Aid09,Aid13}), production constraints (Bouchard and Nguyen Huu, 2013 \cite{Bouchard13}) or gas storage levels (Douglas and Popova, 2008 \cite{Douglas08}). \\

In this article, we propose an analysis where the futures prices can be related to the spot price by arbitrage arguments, independently of the storability properties of the underlying asset. This does not mean, however, that the model is especially designed for non-storable commodities, as done, for example, in Benth \emph{et al} (2007 \cite{Benth_Kallsen_Meyer-Brandis_07} and 2008 \cite{Benth_Saltyte-Benth_Koekebakker_08}), Meyers-Brandis and Tankov (2008 \cite{Meyer-Brandis_Tankov}), and Hess (2013 \cite{Hess_13}), among others. \\
Our framework could be applied to any commodity, would it be storable or not. This can be useful for a large range of industrial companies operating in commodity futures markets, such as energy utilities, airline companies, producers and processors of metals or of agricultural products. The profit maximization of such companies relies on the spot prices as well as on the futures prices. Within our framework, once a spot price model and the market price of demand risk are chosen, there is only one futures price model that would be consistent with their optimization process.\\

The remaining of the article is organized as follows. First, we first propose a simple model of production, storage and trading, where an operator maximizes his expected utility. This operator has no impact on the spot price. The storage costs and the production function are supposed to be convex. The storage capacity is bounded. The same is true for the instantaneous storage and withdrawal. Moreover, the operator has access to a derivative market where a futures written on the commodity is traded. The contract is negociated on one maturity only and the liquidity of the futures market is supposed to be unlimited. Naturally, the operator has no impact on the futures price. \\
Second, we show that the existence of a risk-neutral measure is the consequence of the finiteness of the operator's value function. This result is linked with previous results established in, e.g., Rogers (1994, \cite{rogers}) and Ankirchner and Impeller (2005 \cite{Ankirchner05}) for a pure trader. This result means that if there were no risk-neutral measures, the operator could take advantage of his production capacity or storage facilities to get an infinite utility. In the same section, we also prove that the futures price always converges to the spot price, regardless of the storability properties of the commodity. This result is important, especially for the analysis of electricity markets. It has been indeed pointed out that the electricity futures prices predict realized spot prices rather poorly (Prevot \emph{et al}, 2004 \cite{Prevot04}). This study was however done on monthly contracts while the convergence issue concerns only maturities close to zero. For instance, day-ahead futures contracts quoted on the German electricity market exhibit lower discrepancy with the realized spot prices (see Viehmann, 2011 \cite{Viehmann11}).\\
We finally discuss the trading-production problem faced by the operator, with a specification of the demand dynamics. We obtain an explicit formula for the volatility of the futures contract and we relate it to the volatility of the underlying conditional demand for the commodity. Moreover we argue that, in a Markovian setting and for an agent having a power type utility function, the optimal command for the management of storage is of a bang-bang type. Not surprisingly, the decision to store or to withdraw the commodity is based on the comparison between the spot price and the ratio between the marginal utility of one unit of storage and the marginal utility of the wealth of the operator. Lastly, in this setting the optimal trading strategy on the futures market is such that the operator holds a long position when the futures prices exhibit a positive trend.

\section{The model for the individual producer}
\label{sec:model}

In this section we provide the full description of a simple model of production, storage and trading. This model will be used and investigated throughout the whole article. For the sake of readability, the proofs are all gathered in the Appendix.

Let $(\Omega, (\mathcal F_t)_{t\in [0,T]}, \mathbb P)$ be a filtered probability space satisfying the usual conditions, i.e. $(\mathcal F_t)$ is a $\mathbb P$-completed and right-continuous filtration. Moreover, for mathematical convenience, we assume that $\mathcal F_{T-}=\mathcal F_T$. This property is satisfied, e.g., when $(\mathcal F_t)$ is the natural filtration generated by a multivariate Brownian motion. All processes considered in this article are assumed to be defined in this space and adapted to this filtration.

Our agent is the producer of a commodity. He also trades in the associated derivative market. This producer is a price taker. On the physical market, he has the possibility to sell the whole quantity he produced at a specific date, or less (he then stores a part of the production), or more (he then reduces his stocks). On the derivative market, he has the possibility to buy or to sell a certain amount of a unique futures contract. His choices depend on the price conditions he faces, both on the physical and on the derivative markets. In what follows, we will first focus on the physical market. Then we will expose the trading activity on the derivative market.

\subsection{The profit on the physical market}
\label{subs:profit}
On the physical market, the agent has the possibility to decide how much he produces and to manage his stocks dynamically. 
His instantaneous profit $\pi_t$ can be written as follows: 
\begin{equation}
\pi_t = (q_t -u_t) S_t - c(q_t) - k(X_t)
\label{eq:physicalprofit}
\end{equation}
where: 
\begin{itemize}
\item $q_t$ is the production of the agent,
\item $u_t$ is the amount stored ($u_t>0$) or withdrawn ($u_t<0$),
\item $S_t$ is the spot price of the commodity,
\item $c: \mathbb R_+ \to \mathbb R$ is the production function, with $c(0)=0$,
\item $k : \mathbb R_+ \to \mathbb R$ is the storage function, with $k(0)=0$,
\item $X_t$ is the storage level at time $t$.
\end{itemize}
\begin{assumption}\label{ck} We assume that both functions $c$ and $k$ are differentiable, strictly increasing, strictly convex and nonnegative. \end{assumption}
Moreover, we will always work under the following assumption on the dynamics of the spot price:
\begin{assumption}\label{Scont}
Let $(S_t)$ be a bounded continuous process. 
\end{assumption}

The producer faces several constraints on the physical market. They can be summarized as follows: 
\begin{itemize}
\item the agent's production cannot exceed his capacity $\overline q$ : $q_t \in [0, \overline{q}]$ for some $\overline q >0$;
\item instantaneous storage and withdrawal are bounded, with $u_t \in [\underline{u}, \overline{u}]$ for given thresholds $\underline u < 0 < \overline u $;
\item the storage capacity itself is bounded by some $\overline X \geq 0$, so that adding the positivity constraint on the inventories we have $X_t \in [0,\overline{X}]$ a.s. for all $t\in [0,T]$;
\item the storage dynamics is: $dX_t = u_t dt $, with $X_0 = u_0 > 0$.
\end{itemize}
To simplify, we assume that there is no uncertainty on the production of the commodity. The only source of uncertainty comes from the demand side, which is reasonable for a large number of commodities. 

\subsection{The trading activity on the derivative market}
On the derivative market there is only one futures contract available, for a given maturity $T>0$. The price of this contract at $t$ is $F_t = F(t,T)$. We assume that the futures price process $(F_t)$ is a continuous semi-martingale, adapted to the filtration $(\mathcal F_t )$, and that the interest rate is zero. The value of the trading portfolio on this contract is given by: 
\begin{equation}V^\theta _T = \int_0^T \theta_t dF_t  \label{eqV}\end{equation}
where $\theta$ is any real-valued predictable $(F_t)$-integrable process such that $\theta_{T-} := \lim _{t \uparrow T} \theta_t$ exists a.s.
We consider that the futures market is liquid, which is standard in this context. 
On the empirical point of view, this is the case for a very large spectrum of commodities, especially on the short-term maturities. This is however less obvious, at least up to now, for electricity markets (A\"id, 2015 \cite{Aid15}, Chap. 2, Sect. 2.2.3).
 
\subsection{The production-trading problem}

As a commodity producer, the agent acts so as to maximize the expected utility of his terminal wealth. His utility function $ U: \mathbb R_+ \to [-\infty ,\infty[ $ satisfies Inada conditions and is such that $U(x) \to \infty$ whenever $x\to \infty$. Moreover, we assume Reasonable Asymptotic Elasticity (RAE: see Kramkov and Schachermayer, 1999 \cite{KSchach}):  \[ AE(U) := \limsup_{x\to \infty} \frac{xU'(x)}{U(x)} < 1\]  

In this setting, we propose the following production-trading problem: 
\begin{equation} 
v(r_0) := \sup_{u, q, \theta} \Esp \left[ U\left( r_0 + \int_0^T \pi_t dt + V^\theta _T + \theta_{T-} (F_T-S_T)\right) \right]
 \label{spotmodelThree}
\end{equation}
where:
\begin{itemize}
\item $r_0 >0 $ is the initial wealth of the operator,
\item $\pi_t$ is  the instantaneous profit on the physical market, expressed as a function of the quantities produced and stored, given by Equation (\ref{eq:physicalprofit}),
\item  $V^\theta _T$ is the value of the trading portfolio in the futures market as given by Equation (\ref{eqV}),
\item  the term $\theta_{T-} (F_T-S_T)$ can be explained by the delivery conditions of the futures contract at expiration (see the heuristic discussion for the discrete time case, which follows).
\end{itemize}

The controls $(u,q,\theta)$ have to satisfy the following additional constraints:
\begin{itemize}      
\item constraint on the wealth of the agent to prevent infinite borrowing: 
\begin{equation} R_t ^{r_0 ,u,q,\theta} := r_0 + \int_0 ^t \pi_s ds + V^\theta _t + \theta_{T-} (F_T -S_T) \mathbf 1_{(t=T)}  \geq -a, \quad t\in [0,T] ,\label{eqR} \end{equation}
for some threshold $a>0$,
\item the production-storage controls $(u_t, q_t)_{t\in [0,T]}$ are predictable processes with respect to the filtration $(\mathcal F_t )$ and they satisfy the constraints previously described in Paragraph \ref{subs:profit}.
\end{itemize}

\textbf{Discrete-time heuristics.} Let us now provide some heuristics in discrete-time, in order to better explain the form of the continuous-time problem (\ref{spotmodelThree}). The terminal total wealth for a trader-producer who produces, e.g., energy from fuels, and trades in futures contracts on energy over the finite time grid $\{0,1, \ldots ,T\}$ with $T\in \mathbb N$ can be written as follows: 
\begin{eqnarray*} R_T= \sum_{t=0} ^{T-1} [(q_t -u_t)S_t - c(q_t)-k(X_t)] &+& \sum_{t=0} ^{T-1} \theta_t (F_{t+1}-F_t) \\
&+&  \theta_{T-1}F_{T-1} - h_T S_T - c(q_T) - k(X_T) \end{eqnarray*}
where $h_T$ is the quantity bought or sold at terminal date in order to fulfill the commitment taken on the futures market at the expiration of the contract, i.e. $h_T$ is such that:
\[  \theta_{T-1} = h_T + q_T - u_T\] 
so that the terminal total wealth becomes:
\begin{eqnarray} R_T &=& \sum_{t=0} ^{T} [(q_t -u_t)S_t - c(q_t)-k(X_t)] + \sum_{t=0} ^{T-2} \theta_t (F_{t+1}-F_t)
+ \theta_{T-1} (F_T -S_T)\nonumber \\
&=& \sum_{t=0} ^{T} [(q_t -u_t)S_t - c(q_t)-k(X_t)] + \sum_{t=0} ^{T-1} \theta_t (F_{t+1}-F_t)
+ \theta_{T-1} (F_{T-1} - S_T)  \label{discrete} \end{eqnarray}
which constitutes the discrete-time analogue of the total wealth appearing inside the utility function in Equation (\ref{spotmodelThree}). Notice that we assumed that $F$ is continuous in time, so that in particular $F_{T-}=F_T$ a.s.

In what follows, we prove that a necessary condition for the problem (\ref{spotmodelThree}) to be well-posed is the equality
 $F_T =S_T$, so that the third summand vanishes in Equation (\ref{eqR}). 
Our first objective is to deduce, from the well-posedness of the problem described by Equation (\ref{spotmodelThree}), the no-arbitrage condition that would link spot and futures prices, i.e. to prove that there exists some equivalent probability measure $Q$  such that $$F_t  = \Esp_Q\left[S_T|\mathcal{F}_t\right], \quad t\in [0,T]$$
and to compute this futures price explicitly. 

\section{Existence of the optimum and spot-futures no-arbitrage relations}
\label{sec:existence}

In this section we derive the existence and the uniqueness of an optimal solution $(q^*,u^*,\theta^*)$ for the optimization problem (\ref{spotmodelThree}).\ Meanwhile, we obtain no-arbitrage relations between the spot and futures prices, as well as the convergence of the futures prices to the spot prices when the time-to-maturity goes to zero. Moreover, we show that the optimal production $q^*$ can be computed explicitly even in this general framework. For expository reasons, the explicit expressions for the other optimal quantities (storage and trading activity $(u^*,\theta^*)$) will be given in the next section. We remind the reader that all the proofs can be found in the appendix.  

\begin{assumption}\label{finite} Let $v(r_0) < \infty$ for some initial wealth $r_0>0$.\end{assumption}

\textbf{Convergence of the futures to the spot prices and no-arbitrage relation.} Our first result states that, as long as our optimization problem is well-posed, one must have a convergence of the futures price towards the spot price when the time-to-maturity tends to zero. This is true even when the underlying commodity of the contract is non storable.

\begin{prop}\label{F=S}
Under Assumption \ref{finite}, we have $F_T = S_T$. 
\end{prop}

Before proving the existence of a solution to our optimization problem, we deduce the no-arbitrage property for the futures contracts from the finiteness of $v(r_0)$. This is the content of the next proposition, which adapts arguments from Proposition 1.2 in Ankirchner and Imkeller (2005 \cite{Ankirchner05}), where No Free Lunch with Vanishing Risk (henceforth NFLVR) for simple trading strategies is deduced from the well-posedness of an optimal pure investment problem.  We refer to Delbaen and Schachermayer's article (1994, \cite{DS}) for the definition of NFLVR as well as the proof of their celebrated version of the fundamental theorem of asset pricing. In this article, we deduce something stronger, namely a variant of NFLVR not only for (simple) trading portfolios but also for production and storage. To be more precise, let us redefine the NFLVR condition for our setting. We recall that a simple trading strategy $\theta$ is any linear combination of strategies of the form $\alpha \mathbf 1_{]\tau_1,\tau_2]}$ where $\alpha$ is a bounded $\mathcal F_{\tau_1}$-measurable random variable and $\tau_1$ and $\tau_2$ are $[0,T]$-valued stopping times for the filtration $(\mathcal F_t)$.
\begin{definition}
A \emph{Free Lunch with Vanishing Risk with simple trading strategies, production and storage} is a sequence of admissible plans $(q_t ^n , u_t ^n , \theta_t ^n)$, $n\geq 1$, such that:\begin{enumerate}
\item[(i)] each $\theta^n$ is a simple trading strategy,
\item[(ii)] $R_T ^{0,n} := R_T ^{0,q^n,u^n,\theta^n}$ converges a.s. towards some nonnegative r.v. $R_T ^0$ satisfying $\mathbb P(R_T ^0 >0)>0$ and 
\item[(iii)] $\| (R_T ^n) ^- \|_\infty \to 0$ as $n\to \infty$. 
\end{enumerate}
We will say that NFLVR with simple trading strategies, production and storage is satisfied if there are no such admissible plans in the model.
\end{definition}

Notice that since the production and storage controls are bounded, there exists a constant $M >0$ such that $|\int_0 ^T \pi_t dt | \leq M$ for any admissible $(q,u)$ giving the instantaneous profit $\pi_t$. This fact will be used in the proof of the following result. 

\begin{prop}\label{NA}
Under Assumptions \ref{ck}, \ref{Scont}, \ref{finite} and for all $r_0 > M$ we have that $v(r_0) <\infty$ implies NFLVR with simple trading strategies, production and storage. In particular, NFLVR with simple trading strategies and for futures prices $(F_t)$ also holds. \end{prop}

Notice that the previous result does not necessarily imply the existence of a (local) martingale measure $Q$ for the futures prices $(F_t)$ and hence the well-known formula $F_t = \mathbb E_Q [S_T \mid \mathcal F_t]$. Indeed to have the existence of such a measure, our model should satisfy the classical No-Arbitrage (NA) condition\footnote{See Definition 9.2.8 in \cite{DS}.} too (see, e.g. Theorem 9.7.6 in \cite{DS}). This is something that would need to be imposed later in this paper.

\textbf{Existence and separation principle.} An immediate consequence of the previous convergence result is the following separation principle, stating that solving our optimization problem is equivalent to maximize first with respect to the production control $q$ and then with respect to the storage and trading controls $(u,\theta)$. On the other hand, maximizing the production can also be performed in two steps. Let us denote:
\[ v(r_0) = \sup_{u,q,\theta} \mathbb E\left [  U\left(r_0+Y_T ^{q} + Z_T ^u + V_T^{\theta}\right)  \right]\]
where we set
\[ Y_T ^{q}  := \int_0 ^T (q_t S_t -c(q_t))dt,\quad Z_T ^u := -\int_0 ^T (u_t S_t + k(X_t))dt\]

We can solve our problem in two separate steps. First we solve $v(r_0)$ with respect to the production control $q$ (for given $u,\theta$); second, we solve with respect to the controls $(u,\theta)$. Let us start from the production side.

\begin{prop} \label{production} Under Assumptions \ref{ck}, \ref{Scont}, \ref{finite}, for any given admissible investment strategy $\theta$ and storage policy $u$, the optimal production control $q^*$ is given by
\begin{equation} \label{optqu} q_t ^* = (c')^{-1} (S_t) \vee \bar q ,\quad t\in [0,T] \end{equation} 
\end{prop}

Let us denote
\[ Y_T ^* := Y_T ^{q^*} = \int_0 ^T (q^* _t S_t -c(q^* _t))dt \]
where $q_t ^*$ is given by (\ref{optqu}). Now, let us consider the optimal storage/trading problem
\begin{equation} v(r_0) := \sup_{u,\theta} \mathbb E\left [  U\left(r_0+Y^* _T + Z_T ^u + V_T^{\theta}\right)  \right]\label{storinv}\end{equation}

The next result establishes existence of a unique optimal storage/trading policy $(u^*,\theta^*)$. 

\begin{prop}\label{prop:existence}
Under Assumptions \ref{ck}, \ref{Scont} and \ref{finite}, there exists a unique solution $(u^*,\theta^*)$ to the problem (\ref{storinv}). \end{prop} 

The fact that the trading-production problem above can be solved in successive steps does not mean that the optimal controls are independent. Only the production control $q$ can be deduced independently from $u$ and $\theta$. Indeed, since the producer has no impact on the spot price, his optimal strategy is simply to equal his marginal cost of production with the spot price. Thus, no matter whether there exists a futures market or not, one would observe the same production level $q$.  This is not the case for the optimal storage policy $u$ and the optimal trading strategy $\theta$: they are not independent. This fact has two consequences. First, the introduction of a futures market modifies the way storage capacities are managed. This point may be of interest for the econometric analysis of the relations between the level of the inventories and the prices. Second, as soon as the industrial process includes storage activities, the trading cannot be separated from the storage without a loss of value. This should be taken into account for the organization of the trading activities in industrial companies.

\section{The optimal production-trading problem}
\label{sec:solution}
In this section we illustrate, in a simple setting, how the approach developed in the previous sections can lead to a consistent model for futures prices. \\ 
We discuss the trading-production problem faced by the operator with a specification of the demand dynamics. Then we determine the volatility of the futures contract and we relate it to the volatility of the underlying conditional demand.
\subsection{The dynamics of the demand for the commodity}

 In order to illustrate the previous approach, we use a simple dynamics for the demand. In particular, since our goal is not to propose a model that would match all stylized facts, we do not include features such as seasonality and jumps. 

In commodity markets, the spot price  $S_t$ results from the availability of the raw material on the physical market. In our model, this availability is measured through the confrontation of the total capacities of the market and the demand for the commodity, in the following way:  
 \begin{equation} S_t = b \cdot g(\bar C - D_t) \cdot f(D_t) \label{eqS}\end{equation}
with $b$ a constant of normalization for dimension purposes, $\bar C > 0$ the maximum available production and storage capacities of the market (supposed constant), $D_t$ the total exogenous demand for the commodity (which is an $(\mathcal F_t)$-adapted continuous process) and $f(D)$ the marginal cost of production and storage for a demand level $D$.  \\
As there is a non negativity constraint on inventories, the spot price can jump to very high levels when the total capacities are not sufficient to fully satisfy the demand. This behavior is captured by the scarcity function $g$:
$$g(x) = \mathbf 1_{x>0} \cdot \min (1/x,1/\epsilon) + \mathbf 1_{x<0} 1/\epsilon$$

The effect of scarcity on commodity prices is clearly illustrated for the case of oil in Buyuk\c sahin \emph{et al} (2008 \cite{Buyuksatin08} p. 56, fig. 10). The specific form of $g$ above has been successfully implemented in the case of electricity spot prices in A\"id \emph{et al} (2013 \cite{Aid13}). 

Since the production optimization problem has been solved in Proposition (\ref{production}), it remains to treat:
\[ v(x) =  \sup_{u,\theta} \mathbb E\left [  U\left(x+Y^* _T +Z_T ^u + V_T^{\theta}\right)  \right]\]
where 
\[Y_T ^* = \int_0 ^T (q^* _t S_t -c(q^* _t))dt,  \quad \quad t\in [0,T].\]
with$\quad q_t ^*$ as in (\ref{optqu}).
We recall that $Z_T ^u = u_0 + \int_0 ^T u_t dt$ is the cumulated storage and that  $V_T ^\theta = \int_0 ^T \theta_t dF_t$ is the portfolio traded over the period $[0,T]$.

We assume that the futures price process $(F_t)$ is an It\^o process. 
More precisely:
\begin{assumption}\label{assCDF} 
\begin{enumerate}
\item Let the demand for energy $D_t$ be mean reverting, with a long-run mean set to zero: 
\begin{equation} dD_t = a D_t dt + \sigma dW_t \end{equation}
where $a,\sigma$ are constants and $W$ is a standard Brownian motion. We denote by $(\mathcal F_t )$ the natural filtration generated by $W$ and completed with the $\mathbb P$-null sets.
\item Assume that the futures price F is an It\^o process fulfilling
\[ dF_t = \alpha_t dt + \beta_t dW_t \]
where $\alpha, \beta$ are some $(\mathcal F_t)$-predictable real-valued processes such that a.s.
\[\int_0 ^T |\alpha_t | dt + \mathbb E \left[ \int_0 ^T \beta_t ^2 dt \right] <\infty\]\end{enumerate}
\end{assumption}

The integrability assumption on the volatility is here only to have the (true) martingale property of $F_t$ and consequently the very useful formula $F_t = \mathbb E^Q [S_T \mid \mathcal F_t]$. 

\subsection{Equivalent martingale measures and forward volatility}

First of all, we notice that Assumption \ref{assCDF} together with $F_T = S_T$ (ref. Proposition \ref{F=S}) implies that \[ F_t = \mathbb E^Q [S_T \mid \mathcal F_t], \quad t\in [0,T]\] where $Q$ is an equivalent martingale measure for the futures process $(F_t)$, which means that $Q$ must satisfy:
\[ L_t ^\lambda := \frac{dQ}{d\mathbb P} \mid_{\mathcal F_t} = \exp \left \{-\int_0 ^t \lambda_s dW_s -\frac{1}{2} \int_0 ^t \lambda_s ^2  ds \right \}\]
where $\lambda$ is a $(\mathcal F_t)$-adapted process (viewed as ``market price of demand risk'') such that
\begin{itemize}
\item $\alpha_t - \lambda_t \beta_t =0$ a.e. $d\mathbb P \otimes dt$,
\item $\int_0 ^T \lambda_s ^2  ds < \infty $,
\item $\mathbb E[ L_T ^\lambda ] =1$.
\end{itemize}

At this point, in order to specify completely the dynamics of the futures price under $\mathbb P$, we need to assume a particular and tractable form for the market price of demand risk $\lambda_t$.
 
\begin{assumption}\label{mpdr}
Let us assume that $\lambda_t = \lambda_0 (t) + \lambda_1 (t) D_t$, $t\in [0,T]$, where $\lambda_0 , \lambda_1 : [0,T] \to \mathbb R$ are deterministic functions such that the last three properties above are satisfied.
\end{assumption}

A consequence of this assumption is that the drift  $\alpha_t$  of the futures price takes the form $\alpha_t = (\lambda_0 (t) + \lambda_1 (t) D_t ) \beta_t$, which is completely determined up to the volatility $\beta_t$.

We will see in what follows that the special form of the production function defining the spot price $S_T$ in (\ref{eqS}) implies a particular functional form for the volatility of the futures price process $(F_t)$. 

Let $Q$ be the equivalent martingale measure corresponding to the market price of demand risk $\lambda_t$ as in Assumption \ref{mpdr}. Under such a measure, the demand has a dynamics characterized as follows:
\[ dD_t = ((a + \lambda_1 (t)\sigma) D_t + \lambda_0 (t)\sigma)dt + \sigma dW_t ^Q \]
where $W^Q$ is a standard $Q$-BM. Thus, the conditional distribution of $D_T$ given $D_t$ under $Q$ is Gaussian with conditional mean $m^Q _{t,T}$ and variance $\Sigma^2 _{t,T}$ given by
\begin{eqnarray} \label{mean} 
m^Q _{t,T} &=& e^{\int_t ^T (a+\lambda_1(s)\sigma) ds} \left( D_t + \int_t ^T e^{-\int_0 ^s (a+\lambda_1(u)\sigma) du} \lambda_0 (s) \sigma ds \right), \\ \label{var}
\Sigma_{t,T}^2 &=& \sigma ^2 \int_t ^T e^{-2\int_t ^s (a+\lambda_1 (u)\sigma) du} ds. 
\end{eqnarray}

To complete the description of our model, we set a specific shape for the marginal cost of production.
 
\begin{assumption}\label{margcost} Let the marginal cost of production $f$ be equal to
\[ f (d) = d^\alpha \mathbf 1_{(0\leq d\leq M)} + M^\alpha \mathbf 1_{(d\geq M)} ,\quad d\in \mathbb R\]
for some exponent $\alpha \in ( 0,1)$ and some upper bound $M>0$ such that our conditions on $f$ are fulfilled. Moreover, let $M \geq \bar C -\epsilon$. \end{assumption}

Under all these assumptions, we can express the spot price $S_t$ as a function of the demand $S_t = \psi (D_t)$, where the function $\psi$ is given as follows:
\begin{equation} \psi(d) = b \cdot \left( \frac{d^\alpha}{\epsilon} \mathbf 1_{(0\leq d < \bar C - \epsilon)} + \frac{d^\alpha}{\bar C - d}\mathbf 1_{(\bar C - \epsilon \leq d < M)} + \frac{M}{\bar C - d} \mathbf 1_{(d\geq M)}\right)\label{spot-demand}\end{equation} Notice that the spot price $S_t$ is always nonnegative. 

The futures price at time $t$ computed under the above measure $Q$ is given by
\[ F_t = E_t ^Q [\psi (D_T) ], \quad t \in [0,T]\]
where $E^Q _t$ denotes the conditional $Q$-expectation given $\mathcal F_t = \mathcal F^{D} _t$. 
We denote by $h_{T,D_t}(y)$ the conditional density of $D_T$ given $D_t$, that is: 
\[ h_{T,D_t}(y) = \frac{1}{\Sigma_{t,T}\sqrt{2\pi}} \exp \left(-\frac{(y-m_{t,T}^Q)^2}{2\Sigma_{t,T}^2}\right)\]
where the mean $m^Q _{t,T}$ and the variance $\Sigma^2 _{t,T}$ are given in, respectively, (\ref{mean}) and (\ref{var}). We recall that the variance does not depend on $D_t$.

We can express the futures price $F_t$ as a function of the demand at time $t$, $D_t$, as:
\[ F_t = \varphi (t,D_t) = \int_{\mathbb R} \psi(y) h_{T,D_t}(y) dy\]
A simple application of It\^o's formula 
together with the martingale property of the futures price $F_t$ under $Q$ gives that the volatility of the futures price, $\beta (t,D_t)$, is given by:
\[ \beta_t = \beta_t ^T = \sigma \frac{\partial \varphi}{\partial d} (t,D_t)\]
If we compute explicitly the first and second derivatives of the futures price, $\varphi (t,D_t)$, with respect to the demand, we obtain the following result giving a complete specification of the parameters of the forward dynamics.

\begin{prop}
Under Assumptions \ref{assCDF}, \ref{mpdr} and \ref{margcost}, the well-posedness of the optimal production-trading problem (\ref{spotmodelThree}) implies that
\[ dF_t = \alpha_t dt + \beta_t dW_t\]
where
\begin{eqnarray*} \alpha_t &=& \tilde \alpha (t,D_t) = (\lambda_0 (t) + \lambda_1 (t) D_t) \beta_t \\
\beta_t  &=& \tilde \beta (t,D_t) =  \sigma \int_{\mathbb R} \psi(y) \frac{y-m_{t,T} ^Q}{\Sigma_{t,T}^2} e^{\int_t ^T (a+\lambda_1 (u)\sigma)du } h_{T,D_t} (y) dy  \end{eqnarray*}
for all $t\in [0,T]$. Moreover, the forward volatility $\tilde \beta (t,D_t)$ is increasing in the demand.
\end{prop}

\subsection{The production-trading optimization problem in a Markovian setting} 

Let us now provide, for the sake of completeness, an informal discussion of the optimal solutions within the Markovian model determined in the previous proposition. Let us assume that the preferences of the agent are of power type, i.e. $U(x)= x^\gamma$, $x > 0$, where $\gamma \in(0,1)$. Recall that the problem we want to solve is the following:
\begin{equation} \label{max-power}
v(x) := \sup_{(u, q, \theta) \in \mathcal A} \Esp \left[ \left( r_0 + \int_0^T \pi_t dt + V^\theta _T \right)^\gamma \right]
\end{equation}
where $r_0 >0$ is the initial wealth. Recall that $\pi_t$ is the profit rate given by:
\[ \pi_t = (q_t -u_t) S_t - c(q_t) - k(X_t), \quad X_t = u_0 + \int_0 ^t u_s ds\]
while $V_T ^\theta = \int_0 ^T \theta_t dF_t$ is the gain from the self-financing portfolio traded on the futures market. $\mathcal A$ denotes the set of all admissible controls $(u,q,\theta)$. More precisely, we will say that a triplet $(u,q,\theta)$ is an admissible control if:
\begin{itemize}
\item $q=(q_t)_{t\in [0,T]}$ and $u=(u_t)_{t\in [0,T]}$ are adapted processes with values, respectively, in $[0,\bar q]$ and $[\underline u, \overline u]$ ;
\item $\theta=(\theta_t)_{t\in [0,T]}$ is any predictable real-valued $F$-integrable process such that
the resulting wealth is a.s. nonnegative at any time, i.e.
\[ r_0 + \int_0^t \pi_s ds + V^\theta _t \ge 0, \quad t\in [0,T]\]
\end{itemize}

The relevant state variable of the problem is $Z = (R,X,D)$ where $R$ is the wealth of the agent, i.e. $$ R_t  = r_0 + \int_0^t \pi_s ds + V^\theta _t ,\quad t\in [0,T]$$ 
The dynamics of the state variable is given by:
\begin{align*}
dR_t & = \left[ (q_t - u_t) \psi (D_t) - c(q_t) - k(X_t) + \alpha(t,D_t) \theta_t \right] dt	+ \beta(t,D_t) \theta_t dW_t \\
dX_t & = u_t dt \\
dD_t & = aD_t dt + \sigma dW_t 
\end{align*}
Let us introduce the value function of the optimization problem as 
\[ v(t,r,x,d) = \sup_{(u, q, \theta) \in \mathcal A_t } \Esp \left[ \left( R_T \right)^\gamma | Z_t =(r,x,d) \right]\]
where $\mathcal A_t$ denotes the set of all admissible controls starting at time $t$. The corresponding Hamilton-Jacobi-Bellman equation (hereafter HJB equation) is given by
\begin{equation}\label{HJB} -v_t - \sup_{(u,q,\theta) \in A} \L v = 0, \quad \textrm{with } A := [\underline u , \overline u] \times [0,\bar q] \times \mathbb R\end{equation}
with terminal condition \begin{equation}\label{terminal} v(T,r,x,d) = r^\gamma \end{equation}
and where 
\[ \L v  = u v_x + ad v_d + \left[ (q - u) \psi (d) - c(q) - k(x) + \alpha (t,d) \theta \right]  v_r
		     + \frac{1}{2} \sigma^2 v_{dd} + \frac{1}{2} \beta (t,d) ^2 \theta^2 v_{rr} \]

The HJB equation may be rewritten as
\begin{align*}
0 = - & v_t - ad v_d + k(x) v_r  - \frac{1}{2} \sigma^2 v_{dd} \\
 - & \sup_{(u,q,\theta)\in A} \left\{ u v_x  + \left[ (q - u) \psi(d) - c(q) + \alpha (t,d) \theta \right]  v_r  + \frac{1}{2} \beta (t,d) ^2 \theta^2  v_{rr} \right\}
\end{align*}
Rearranging the terms gives
\begin{align}
0 = - & v_t - ad v_d + k(x) v_r  - \frac{1}{2} \sigma^2 v_{dd} \nonumber  \\
 - & \sup_{(u,q,\theta)\in A} \left\{  (v_x - \psi(d) v_r ) u  - (c(q) - q \psi(d)) v_r + \alpha (t,d) \theta v_r + \frac{1}{2} \beta (t,d) ^2 \theta^2 v_{rr}  \right\}  \label{hjb}
\end{align}

We notice immediately from the above HJB equation that the optimal candidate rule for the storage management is
\begin{equation}\label{u*} u_t ^* = \underline u \mathbf 1_{(\psi(D_t) v_r > v_x)} + \overline u  \mathbf 1_{(\psi(D_t) v_r \leq v_x)}\end{equation}
where, to simplify the notation, we dropped the arguments $(t,R_t,X_t, D_t)$ from the derivatives of the value functions $v_r$ and $v_x$. 
Depending on the ratio $\eta = v_x/\psi(d)$ between the marginal utility of one unit of storage and the spot price, it is optimal either to buy and store at maximum capacity or to withdraw and sell at maximum capacity. One may have thought that this ratio should simply be compared to one. This is not the case as it has to be compared to the marginal utility of one unit of wealth, $v_r$. As pointed at the end of the preceding section, we see that the storage control depends on the wealth and hence on the trading activities.

Furthermore, the heuristic computation above confirms the fact that the optimal control for production is to produce until the marginal cost of production equals the spot price, i.e. Equation (\ref{optqu}).

Finally, solving the maximisation problem in (\ref{hjb}) gives the optimal control for the trading portfolio as
\begin{equation} \label{theta*} \theta^* _t  = - \frac{\alpha(t,D_t)}{\beta (t,D_t)^2} \frac{v_r}{v_{rr}} (t,R_t, X_t , D_t)\end{equation}

Since it is likely that $v_{rr}$ is negative because $U$ is concave, one recovers the expected result that the agent holds a long position if futures prices are exhibiting a positive trend. Moreover, $\theta^*$ is similar to the Sharpe ratio, a tradeoff between the expected trend of the futures prices compared to their volatility.

Notice that, to rigorously solve the optimization problem, one should prove that the Cauchy problem (\ref{HJB}, \ref{terminal}) admits a unique solution with the required regularity together with a verification theorem. This could be achieved using the techniques developed in, e.g., Pham (2002 \cite{pham02}) for quite a large class of multidimensional stochastic volatility models. Adapting such a method to our setting would however go far beyond the scope of the present article. 
\section{Conclusion}
\label{sec:conclusion}
In this article, we develop a parsimonious structural model of commodity prices that can explain the relation between the spot and the futures prices by arbitrage arguments. This result is obtained for every underlying asset, would it be storable or not. We show that the existence of a risk-neutral measure is the consequence of the finiteness of the operator's value function, and that the futures price converges to the spot price, regardless of the storability properties of the commodity. Finally, we discuss the solution of the trading-production problem faced by the agent, with a specification of the demand dynamics. We show how the different controls can be separated. In particular, the optimal storage policy is impacted by the possibility of trading on the futures market. 

\bibliographystyle{plain}

\appendix

\section{Proofs}
\subsection{Proposition on convergence}

\begin{proof}[Proof of Proposition \ref{F=S}] Assume that $\mathbb P(F_T \neq S_T)>0$ and let $A=\{F_T > S_T\}$ and $B=\{F_T < S_T\}$. Consider the following sequence of trading-production strategies: for $n \ge 1$,
\[ q=u=0, \quad \theta^n _t := \left(\alpha \mathbb P\left(A \mid \mathcal F_{T-\frac{1}{n}}\right) -\beta \mathbb P\left(B \mid \mathcal F_{T-\frac{1}{n}}\right)\right) \mathbf 1_{\left(T-\frac{1}{n} \le t \le T\right)},\]
where $\alpha$ and $\beta$ are arbitrary positive numbers. Since $A$ and $B$ are $\mathcal F_{T-}$-measurable ($S_t$ and $F_t$ are both continuous processes), each $\theta^n$ is a predictable and $(F_t)$-integrable trading strategy. Moreover the left-limit $\theta_{T-}$ exists for all $n\ge 1$. Pursuing such a sequence of strategies yields a limiting terminal wealth as $n \to \infty$ given by $r_0 + \alpha (F_T -S_T)\mathbf 1_A + \beta (S_T -F_T)\mathbf 1_B$.

Hence, letting $\alpha \to \infty$ and $\beta =0$, if $\mathbb P(A) >0$ or $\beta \to \infty$ and $\alpha =0$ if $\mathbb P(B)>0$ we get $v(x) = \infty$ (recall that $U(x) \to \infty$ when $x\to \infty$), which contradicts the well-posedness of our maximization problem. Thus, we can conclude that a.s. $F_T = S_T$. \end{proof}

For any family $\mathcal X$ of random variables, $\textrm{conv}(\mathcal X)$ will denote the set of all convex linear combinations of elements in $\mathcal X$.

\subsection{Proposition on $v(r_{0})$ }

\begin{proof}[Proof of Proposition \ref{NA}]
First notice that for all $r_0 >M$ we have
\[ \infty > v(r_0) = \sup_{u,q,\theta} \mathbb E[U(R_T ^{r_0,u,q,\theta})] \le \sup_\theta \mathbb E[ U(r_0 + V_T ^\theta)] \le \sup_\theta \mathbb E[ U(M + V_T ^\theta)] =: v_I (M),\]
where $v_I$ denotes the value function for the pure investment optimization problem. Now, suppose that NFLVR with simple trading strategies, production and storage is violated, so that we can find a sequence of terminal payoffs $R_T ^n = \int_0 ^T \pi^n _t dt + \int_0 ^T \theta^n _t dF_t$ such that $R_T ^n \to R_T ^0$ for some nonnegative random variable $R_T ^0$ with $\mathbb P(R_T ^0 >0)>0$ and $\| (R_T ^n)^- \|_\infty \to 0$ as $n\to \infty$. Here $\pi^n$ denotes the instantaneous profit coming from a production $q^n$ and a storage $u^n$. Hence, we have
\[ R_T ^n -M \le V_T ^n \le R_T ^n +M,\]
where we denote $V^n := V^{\theta^n}$. By Theorem 15.4.10 in \cite{DSbook} there exists a sequence $\tilde V_T ^n \in \textrm{conv}(V_T ^n , V_T ^{n+1}, \ldots)$ which converges a.s. to some random variable $\tilde V_T ^0$, which takes values in $[R^0 _T -M, R_T ^0 +M]$ a.s. Therefore $\tilde V_T ^0 \ge R_T ^0 -M \ge -M$ and $\mathbb P(\tilde V_T ^0 > -M) >0$. Moreover, since $\|(\tilde V_T ^n +M)^- \|_{\infty} \le \| (R_T ^n)^- \|_{\infty} \to 0$ and the latter converges to zero as $n \to \infty$, we also have $\|(\tilde V_T ^n +M)^- \| \to 0$. Hence, using Proposition 1.2 in \cite{Ankirchner05}, we obtain that $v_I (M) =\infty$ implying $v(r_0) = \infty$.
\end{proof}

\subsection{Proposition on the optimal production $q^*$}
\begin{proof}[Proof of Proposition \ref{production}]
It suffices to maximize $\omega$-wise inside the integral in the term $Y_T^{q}$ containing the production controls. Differentiating with respect to $q_t$ for a fixed $t$ gives $S_t - c'(q_t) =0$ so that, taking into account the constraint $q_t \in [0,\overline q]$ and since $c$ is strictly convex, we have (\ref{optqu}).
\end{proof}

\subsection{Proposition on the optimal storage and trading portfolio $(u^*,\theta^*)$ }
\begin{proof}[Proof of Proposition \ref{prop:existence}]
First of all, if one admits the existence of a solution, its uniqueness follows at once from the strict concavity of the utility function $U$. Let $(u^n,\theta^n)$ be a maximizing admissible sequence for the problem
(\ref{storinv}), i.e. $\mathbb E\left [  U\left(r_0 +Y^* _T + Z_T ^n + V_T^n \right)  \right] \to v(r_0 )$ as $n\to \infty$, where we denoted
\[ Z_T ^n := -\int_0 ^T (u^n _t S_t + k(X^n _t))dt, \quad X_t ^n := u_0 + \int_0 ^t u^n _s ds, \quad V_T ^n := \int_0 ^T \theta^n dF_t .\]

We prove the compactness property of the sequences $u^n$ and $\theta^n$ separately. \\
For the sequence of storage strategies $u^n$, we use the Koml\'os theorem, stating that for any sequence of r.v.'s $(\xi^n)$ bounded in $L^1$, one can extract a subsequence $(\xi^{n_k})$ converging a.s. in Cesaro sense to a random variable $\xi^0 \in L^1$ (see, e.g., Theorem 5.2 in Kabanov and Safarian, 2009 \cite{KS09}). We apply this theorem to the sequence of processes $u^n$, that can be viewed as random variables defined on the product space $(\Omega \times [0,T], \mathcal P , d\mathbb P dt)$ where $\mathcal P$ is the predictable $\sigma$-field. The sequence $u^n$ is clearly in $L^1$ since it takes values in the interval $[-\overline u , \overline u]$. Thus, there exists a predictable process $u^0$ taking values in the same interval, such that the Cesaro mean sequence $\tilde u^n := (1/n)\sum_{j=1}^n u^j$ converges a.e. towards $u^0$. Indeed it is immediate to check that the sequence $\tilde u^n$ takes values in $[-\overline u , \overline u]$ as well. Moreover, the cumulated storage process along the new sequence, $\tilde X^n _t := \int_0 ^t \tilde u_s ^n ds$, is well-defined since each $\tilde u^n$ is bounded and it takes values in $[0,\overline X]$. By Lebesgue dominated convergence we have $\tilde X^n _t \to X^0 _t := \int_0 ^t u_s ^0 ds$ a.s. for all $t\in [0,T]$. Since the function $k$ is continuous, we have $k(\tilde X^n _t) \to k(X_t ^0)$ a.s. for all $t$. Finally, thanks once more to the boundedness of the controls and to the continuity of $k$, we have $ \vert \tilde Z^n _T\vert \leq C \vert\int_0 ^T S_t dt \vert $, which is bounded (since $S_t$ is bounded uniformly in $t$). Therefore, applying the dominated convergence theorem again, we get $\tilde Z_T ^n \to Z_T ^0$ a.s. as $n\to \infty$.\\
 As for the compactness of the sequence of trading strategies $\theta^n$, we can work with the corresponding wealth process Cesaro mean sequence, that we denote by $\tilde V_T ^n$. The admissibility property and the uniform boundedness of $\tilde Z_T ^n$ yields that this sequence is uniformly bounded from below by some constant. Therefore, we can apply Theorem 15.4.10 in \cite{DSbook}, implying that there exists a convex combination $\widehat V_T ^n \in \textrm{conv}(\tilde V_T ^n , \tilde V_T ^{n+1}, \ldots)$, which converges a.s. and whose limit is dominated by some $V^0 _T := \int_0 ^T \theta^0 _t dF_t$ for some admissible $\theta^0$. Moreover, applying this procedure to the Cesaro means of storage strategies $\tilde u^n$ one gets another sequence of admissible storage strategies $\widehat u^n$ converging a.s. to the same process $u^0$ as $\tilde u^n$ before. \\

To conclude the proof, we need to show that:
\[ v(r_0) \le  \mathbb E [U(x+Y_T ^* + V_T ^0 + Z_T ^0)]\]
To do so, it suffices to use the assumption that $U$ satisfies RAE by proceeding as in the proof of, e.g., Theorem 7.3.4 in Pham (2000 \cite{PhamBook}). Repeating his arguments gives us the inequality above getting that $(u^0, \theta^0)$ is the optimal storage control $(u^*,\theta^*)$. The proof of existence is now completed.
\end{proof}

\end{document}